\documentclass[pra,letterpaper,showpacs,showkeys,nofootinbib,twocolumn,superscriptaddress]{revtex4}

\usepackage{graphicx}

\usepackage{amsthm}
\usepackage{amsmath}
\usepackage{amsfonts}

\newtheorem{Theorem}{Theorem}
\newtheorem{Def}{Definition}
\newtheorem{Lemma}{Lemma}

\newtheorem{Assumption}{Assumption}

\begin{document}


\newcommand{\Ket}[1]{\ensuremath{\left | #1 \right \rangle}}
\newcommand{\Bra}[1]{\ensuremath{\left \langle #1 \right |}}
\newcommand{\BraKet}[2]{\ensuremath{\left \langle #1 \right | \left. #2 \right \rangle}}
\newcommand{\Tr}[1]{\ensuremath{\mbox{Tr} \left ( #1 \right )}}
\newcommand{\PTr}[2]{\ensuremath{\mbox{Tr}_{#1} \left ( #2 \right )}}


\setlength{\unitlength}{1cm}


\title{The de Finetti theorem for test spaces} 

\author{Jonathan Barrett}
\email{j.barrett@bristol.ac.uk}
\affiliation{H.~H.~Wills Physics Laboratory, University of Bristol, Tyndall Avenue, Bristol
BS8 1TL, U.K.}

\author{Matthew Leifer}
\email{matt@mattleifer.info}
\affiliation{Institute for Quantum Computing, University of Waterloo, 200 University Avenue West, Waterloo, Ontario, Canada, N2L 3G1}
\affiliation{Perimeter Institute for Theoretical Physics, 31 Caroline Street North, Waterloo, Ontario, Canada, N2L 2Y5}

\date{\today}




\begin{abstract} 
We prove a de Finetti theorem for exchangeable sequences of states on \emph{test spaces}, where a test space is a generalization of the sample space of classical probability theory and the Hilbert space of quantum theory. The standard classical and quantum de Finetti theorems are obtained as special cases. By working in a test space framework, the common features that are responsible for the existence of these theorems are elucidated. In addition, the test space framework is general enough to imply a de Finetti theorem for classical processes. We conclude by discussing the ways in which our assumptions may fail, leading to probabilistic models that do not have a de Finetti theorem. 
\end{abstract}

\maketitle


\section{Introduction}
\label{Intro}

There are many scenarios involving probabilistic reasoning about a large number of systems, where it is important that these systems can be regarded as identical and independent. Classical parameter estimation and quantum tomography, in which a source is calibrated by measuring separately a number of systems it has produced, provide excellent examples. More generally, much of experimental science involves repetition of an experiment, with conclusions drawn from observed relative frequencies, where the conclusions are only valid under the assumptions of identity and independence of the separate trials. A problem that arises, therefore, is what justifies these assumptions? De Finetti theorems are designed to answer this question.

De finetti theorems have a fundamental significance for the subjective Bayesian interpretation of probabilities. They enable the subjective Bayesian to explain why a rational agent treats a sequence of trials as identical and independent, without appealing to a notion of objective chance, or to unknown probabilities. In addition to this fundamental significance, they are also important technical tools. For example, quantum de Finetti theorems have been applied to various problems in quantum information theory, including proofs of the security of quantum key distribution \cite{Renner:2005aa}, algorithms for deciding the separability of bipartite quantum states \cite{Doherty:2004, Ioannou:2007} and global optimization of the maximum output purity of quantum channels \cite{Audenaert:2004}.

In this paper, we provide a de Finetti theorem for nonsignalling states on \emph{test spaces}. A test space is a generalization of the sample space of classical probability theory and the Hilbert space of quantum theory. The de Finetti theorem for test spaces includes the classical and quantum theorems as special cases. It also implies a de Finetti theorem for classical processes, which can also be viewed as a theorem about states in theories that exhibit ``superquantum'' correlations, as have been studied recently in quantum information \cite{Barrett:2005a,Barrett:2005aa,Barrett:2005b,Brassard:2005,Broadbent:2005,Buhrman:2005,Dam:2000,Dam:2005,Jones:2005,Khalfi:1985,Popescu:1994,Short:2005,Short:2005a}.

The remainder of this paper is structured as follows.  In \S\ref{CQReview}, the classical and quantum de Finetti theorems are reviewed and their relevance for Bayesian statistics is discussed.  In \S\ref{Test}, the test space framework is introduced and its connection to the convex sets framework used in \cite{Barnum:2006, Barnum:2006a} is explained.  \S\ref{MainRes} states the de Finetti theorem for test spaces and \S\ref{Cons} outlines some of its consequences, including the classical and quantum theorems, and the theorem for classical processes.  \S\ref{Not} discusses the role of the various assumptions of the test space framework and outlines some more general scenarios in which the theorem fails.  \S\ref{Conc} concludes.

\section{The classical and quantum de Finetti theorems.}
\label{CQReview}

De Finetti introduced his theorem in the context of a subjective Bayesian approach to probability theory \cite{Finetti:1993, Finetti:1990aa}. In this approach, probabilities are not defined as limiting relative frequencies, nor as objective properties of the physical world. Instead they are measures of the degrees of belief of a decision making agent. An immediate question is, why should degrees of belief be represented by real numbers obeying the Kolmogorov axioms? De Finetti's answer is provided by his famous Dutch Book argument \cite{Finetti:1993, Finetti:1990aa}\footnote{More sophisticated arguments, framed in terms of decision theory, are advocated by many contemporary subjective Bayesians \cite{Savage:1972aa, Bernardo:2000aa}.}. But a second question is, why do the usual rules of statistical inference apply to these quantities? In particular, how and why should relative frequencies be used to update probability assignments? 

Consider an experiment in which a trial with $d$ possible outcomes is repeated $n$ times. Under a standard sort of analysis, the trials are first judged to be independent and identically distributed, so that the probability of getting the outcome sequence $x_1,\ldots,x_n$ is given by $P^n(x_1,\ldots,x_n) = p(x_1)\times\cdots\times p(x_n)$, for some ``unknown'' probability distribution $p$. The distribution $p$ is a parameter to be estimated. 

To the subjective Bayesian, however, this is problematic since no sense can be given to an ``unknown'' probability. De Finetti provides an alternative analysis. If the trial can in principle be repeated an arbitrary number of times, then the joint distribution over outcome sequences $P^n$ should be defined for any $n$, so consider an infinite sequence of distributions $P^1,P^2,\ldots$  Suppose that for the first $n$ trials, the agent is indifferent as to whether any further trials are actually performed or not, and that the agent is also indifferent as to the order in which the outcomes are reported. This suggests that the sequence $P^1,P^2,\ldots$ should satisfy the following. 

\begin{Def}
$P^n$ is \emph{symmetric} if and only if it is invariant under permutations of the $n$ tests. That is, $P^n(x_1,\ldots,x_n) = P^n(x_{\pi(1)},\ldots,x_{\pi(n)})$ for all permutations $\pi:\{1,2,\ldots,n\}\rightarrow \{1,2,\ldots,n\}$.
\end{Def}

\begin{Def}
The sequence $P^1,P^2,\ldots$ is \emph{exchangeable} if and only if 
\begin{enumerate}
\item $\forall n$, $P^n$ is symmetric, 
\item $\forall n$, $\forall x_1,\ldots,x_n$, \\ $P^n(x_1,\ldots,x_n)=\sum_{x_{n+1}=1}^d \ P^{n+1}(x_1,\ldots,x_n,x_{n+1})$.
\end{enumerate}
\end{Def}

\begin{Theorem}[de Finetti's representation theorem \cite{Finetti:1990aa,Hewitt:1955}]\label{classicaldefinetti}
If the sequence $P^1,P^2,\ldots$ is exchangeable, then $P^n$ can be written in the form
\begin{equation}
\label{indep}
P^n(x_1,\ldots,x_n) = \int_{\Delta_d} \!\mathrm{d}\mu(p) \ p(x_1) \cdots p(x_n),
\end{equation}
where $\Delta_d$ is the set of all probability distributions over the outcomes $\{1,\ldots,d\}$, $\mu$ is a probability measure on $\Delta_d$, $\mu$ is independent of $n$, and $\mu$ is unique. 
\end{Theorem}

This is the classical de Finetti theorem for infinite sequences and a finite number of outcomes. It shows that if the sequence $P^1,P^2,\ldots$ is exchangeable, then $P^n$ can be written \emph{as if} it were generated via a probability distribution $\mu$ over unknown probabilities $p$. In particular, if the first $m<n$ trials are performed, and standard Bayesian updating applied directly to the joint distribution $P^n$, one finds that the posterior probability for the remaining $n-m$ trials is given by 
\begin{align}
P^n(x_{m+1},&\ldots,x_n|x_1,\ldots,x_m) = \nonumber\\ 
& \int_{\Delta_d}\!\mathrm{d}\mu(p|x_1,\ldots,x_m)\ p(x_{m+1})\cdots p(x_n),
\end{align}
where $\mu$ is updated as if Bayesian conditioning had been performed on an unknown parameter $p$ directly:
\[
\mathrm{d}\mu(p|x_1,\ldots,x_m) = \frac{\mathrm{d}\mu(p)\times p(x_1)\times\cdots\times p(x_m)}{P^n(x_1,\ldots,x_m)}.
\]

The quantum de Finetti theorem is a generalization of the classical theorem. It was first presented in Refs.~\cite{Hudson:1976aa, Hudson:1981aa}, and a simpler proof given in Refs.~\cite{Caves:2002aa, Fuchs:2004aa}. Whereas the classical theorem concerned the outcome probabilities of a test that could be repeated an arbitrarily large number of times, the quantum theorem concerns the joint state of an arbitrarily large number of quantum systems. Suppose that each of these systems is associated with a $d$-dimensional Hilbert space $H_d$. The joint state of $n$ systems is then a density operator on the tensor product Hilbert space $H_d^{\otimes n}$. Exchangeability is defined for a sequence of states $\omega^1,\omega^2,\ldots$, where $\omega^n$ is a state of $n$ systems.

\begin{Def}
$\omega^n$ is \emph{symmetric} if and only if it is invariant under permutations of the $n$ systems, i.e.,
\[
\mathrm{Tr}\left(Q_1\otimes\cdots\otimes Q_n \, \omega^n \right) = \mathrm{Tr}\left(Q_{\pi(1)}\otimes\cdots\otimes Q_{\pi(n)}\,\omega^n \right),
\]
for all permutations $\pi:\{1,2,\ldots,n\}\rightarrow \{1,2,\ldots,n\}$, and for any projection operators $Q_1,\ldots,Q_n$.
\end{Def}
Note that this is equivalent to requiring that $\omega^n = S_{\pi} \omega^n S_{\pi}^\dagger$ for all permutations $\pi$, where $S_{\pi}$ is the operator that permutes the $n$ systems according to $\pi$.
\begin{Def}
The sequence $\omega^1,\omega^2,\ldots$ is \emph{exchangeable} if and only if 
\begin{enumerate}
\item $\forall n$, $\omega^n$ is symmetric, 
\item $\forall n$, $\omega^n = \mathrm{Tr}_{n+1} (\omega^{n+1})$, where $\mathrm{Tr}_{n+1}$ is the partial trace over the $n+1$st system.
\end{enumerate}
\end{Def}

\begin{Theorem}[the quantum de Finetti theorem]\label{quantumdefinetti}
If the sequence $\omega^1,\omega^2,\ldots$ is exchangeable, then $\omega^n$ can be written in the form
\begin{equation}\label{quantumindep}
\omega^n = \int_{\Omega} \!\mathrm{d}\mu(\omega) \ \omega\otimes\cdots\otimes \omega,
\end{equation}
where $\Omega$ is the set of density operators on $H_d$, $\mu$ is a probability measure on $\Omega$, $\mu$ is independent of $n$, and $\mu$ is unique.
\end{Theorem}   

From a fundamental point of view, the quantum de Finetti theorem is particularly important for those approaches to quantum theory that take a subjective view of the quantum state \cite{Caves:2002ab, Fuchs:2002aa, Fuchs:2003aa, Pitowsky:2003aa, Caves:2006}. These approaches are closely related to the Bayesian view of probabilities. A quantum state is taken to represent the degrees of belief of an agent, where these might be beliefs about the potential outcomes of measurements. In this case, the notion of an unknown quantum state, so prevalent in the literature, becomes problematic. The quantum de Finetti theorem shows how to dispense with this notion, at least in some situations.

Both of the theorems just presented assume finite sample spaces (or finite dimensional Hilbert spaces), and concern a number of trials or systems that tends to infinity. Both have been generalized in a number of ways. Classical theorems for a finite number of trials are discussed in \cite{Diaconis:1977aa, Diaconis:1980aa, Kendall:1967} and quantum theorems for a finite number of systems in \cite{Koenig:2005aa, Renner:2005aa, DCruz:2007, Renner:2007, Koenig:2007}. In addition to the quantum de Finetti theorem, there is also a de Finetti theorem for quantum operations \cite{Fuchs:2004ab, Fuchs:2004aa}, for representations of unitary groups \cite{Christandl:2006aa, Koenig:2007} and for unitarily invariant quantum states \cite{Mitchison:2007}.

\section{Test spaces}
\label{Test}

The classical de Finetti theorem involves probabilities of outcome sequences for a test that can in principle be repeated an arbitrarily large number of times. The quantum de Finetti theorem involves the joint quantum state of a number of quantum systems that can in principle be arbitrarily large. Both of these can be viewed as special cases of a more general scenario. Suppose that a single system is associated with a number of different possible tests, which are mutually exclusive in the sense that only one can be performed at a time. The classical case is recovered when there is in fact only one such test, and the quantum case when the tests correspond to the different possible measurements on a quantum system. A state is an assignment of probabilities to the outcomes of all the different possible tests. 

Assume further that given $n$ systems, a test for each system can be independently chosen, and that a joint state is an assignment of probabilities to outcome sequences for each possible sequence of tests. Given all this, it is possible to define exchangeability for a sequence of states and to prove a de Finetti representation theorem. These ideas are formalized in this section, with the representation theorem given in the next. 

\subsection{Single systems}
\label{Test:Single}

The technical notion that we use to describe a single system is that of a \emph{test space}. Test spaces were introduced with the explicit purpose of describing probabilistic models more general than classical and quantum theory, but including both as special cases \cite{Foulis:1972, Randall:1973}. 
\begin{Def}
A \emph{test space} consists of a pair $(E,S)$, with $E$ a set, and $S$ a set of countable subsets of $E$ that covers $E$. 
\end{Def}
The idea is that each element of $E$ is a possible outcome of a test. Each set $s\in S$ corresponds to a possible test, with the elements of $s$ being the outcomes of that test. The sets in $S$ may overlap, thus the definition of a test space is designed to allow for the possibility that outcomes of two different tests are identified. The sets $E$ and $S$ themselves may have any cardinality, but the definition stipulates that the outcomes of any particular test are countable.  

For finite $E$, a test space can be conveniently summarized by a Greechie diagram \cite{Greechie:1969}.  Each element of $E$ is represented by a circle, and tests are represented by connecting the corresponding set of circles with a continuous line.  Examples of Greechie diagrams are given in Figs. \ref{Greechie}-\ref{BoxGreechie}.

\begin{figure}
\includegraphics[scale=0.7]{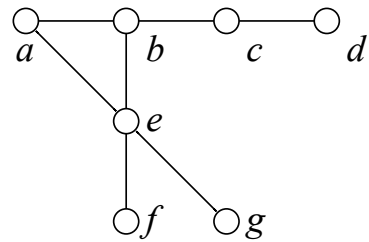}
\caption{A Greechie diagram of the test space $(\{a,b,c,d,e,f,g\},\{\{a,b,c,d\},\{a,e,g\},\{b,e,f\}\})$.\label{Greechie}}
\end{figure}
\begin{Def}\label{states}
A \emph{state} on a test space is a map $\omega: E\rightarrow [0,1]$ satisfying $\sum_{e\in s} \,\omega(e) = 1 \ \forall s\in S$.
\end{Def} 

A state defines probabilities for the outcomes of each test such that (i) these probabilities sum to $1$ for each test, and (ii) if an outcome appears in more than one test it gets the same probability in each case. Given a test space $\mathcal{A}$, write the set of all possible states $\Omega(\mathcal{A})$. Note that it is easy to construct test spaces for which $\Omega(\mathcal{A})$ is the empty set, or for which a particular outcome has probability $0$ in all states, or which may for similar reasons be judged unsatisfactory. Extra assumptions would rule these out, but here there is no need.   

The set of possible state spaces of a test space is generic, in the sense that every finite dimensional convex set arises as the set of states for some test space.\footnote{This statement extends to infinite dimensions. Precisely:
every convex subset of a locally convex topological vector space is affinely homeomorphic to the set of all states on some test space. In \cite{Shultz:1974} this was proved for the state space of an orthomodular lattice, but the set of all finite ortho-partitions of unity of the lattice is a test space that has the same state space.}   
As noted above, discrete classical probability theory is recovered when there is only one test, i.e., when $\mathcal{A}$ is of the form $(E,\{E\})$ (see Fig.~\ref{ClassGreechie}). 
\begin{figure}
\includegraphics[scale=0.7]{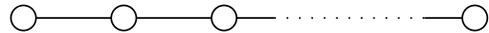}
\caption{Greechie diagram for classical probability theory over a finite set.\label{ClassGreechie}}
\end{figure}
It is also possible to recover classical probability theory over an arbitrary measurable set \cite{Wilce:2000}. Quantum theory with projective measurements is recovered when $\mathcal{A}=(\mathcal{P}(H),\mathcal{M}(H))$, with $\mathcal{P}(H)$ the set of projection operators on a Hilbert space $H$ and $\mathcal{M}(H)$ the set of projective decompositions of the identity. In this case, Gleason's theorem \cite{Gleason:1957} implies that each state corresponds to a density operator $\rho$, with the probability assigned to projector $P$ given by $\mathrm{Tr}(\rho P)$\footnote{This holds provided the dimension of the Hilbert space is $\geq 3$.  A generalization of Gleason's theorem to encompass positive operator valued (POV) measurements does hold for dimension 2 \cite{Busch:2003, Caves:2004}, but the formalism of test spaces is not general enough to encompass these measurements. This is because a POV decomposition of the identity can contain multiple instances of the same term, such as $\{I/2,I/2\}$, and states are constrained to assign the same probability to each instance. In the test space formalism, different outcomes of the same test are always considered distinct. For precisely this reason, generalizations of test spaces known as \emph{effect-test spaces} have been studied that do encompass quantum POV measurements \cite{Pulmannova:1995, Gudder:1997}. However, there is no real need to consider them here because we are primarily concerned with properties of states, and the set of possible state spaces of an effect-test space is no more general than that of a test space.}.  

It will be useful to define arbitrary linear combinations of states. Given states $\omega_1,\ldots,\omega_k \in \Omega(\mathcal{A})$, the linear combination $v = \sum_i \,r_i \omega_i$, for real $r_i$, is defined as a map $E\rightarrow \Re$ that satisfies
\begin{equation}
v(e) = \sum_i\, r_i \omega_i(e) \quad \forall e\in E.
\end{equation}
The set of all linear combinations of states is a vector space denoted $V(\mathcal{A})$. In the case of quantum theory, for example, $V(\mathcal{A})$ is the real vector space of Hermitian operators on the Hilbert space. Clearly, $\Omega(\mathcal{A})$ is a convex subset of $V(\mathcal{A})$. Since $\Omega(\mathcal{A})$ by definition spans $V(\mathcal{A})$, they have equal dimension. Importantly, from hereon we assume the following.

\begin{Assumption}
$V(\mathcal{A})$ is finite dimensional.
\end{Assumption}

Finally, let $V^*(\mathcal{A})$ be the vector space dual to $V(\mathcal{A})$, that is the set of all linear maps $V(\mathcal{A})\rightarrow \Re$. Note that each outcome of a test, that is each $e\in E$, can be uniquely identified with a map $\tilde{e}\in V^*(\mathcal{A})$ such that $\tilde{e}(\omega) =\omega(e)\ \forall \omega\in\Omega(\mathcal{A})$. Under this identification one can write $\omega(e)$ and $e(\omega)$ interchangeably, according to whether states are viewed as assigning probabilities to outcomes of tests, or vice versa. The set $E$ can be viewed as a subset of $V^*(\mathcal{A})$, and it is easy to see that the span of $E$ is equal to $V^*(\mathcal{A})$.

Before moving on to composite systems, we briefly note that there is an alternative approach to operational probabilistic theories that has recently been used to investigate the information theoretic properties of such theories \cite{Barnum:2006,Barnum:2006a,Barrett:2005aa}.  In this approach one starts with a compact convex set $\Omega$, to be interpreted as a space of states, and defines measurement outcomes to be the set of affine functionals $f:\Omega \rightarrow [0,1]$.  The present work could easily have been formulated in this framework, but it would be odd to do so from a subjective Bayesian point of view.  If the states are supposed to represent degrees of belief then it makes sense to start with the objects that they are degrees of belief about, i.e. the tests, rather than the states themselves.  There is no loss of generality in working with test spaces, since the result of \cite{Shultz:1974} implies that any compact convex set can arise as the state space of a test space in the finite dimensional case.

\subsection{Multi-partite systems}
\label{Test:Multi}

In order to consider multi-partite systems, one needs to consider the composition of test spaces and the definition of joint states. Suppose that two systems $A$ and $B$ are associated with test spaces $\mathcal{A}=(E,S)$ and $\mathcal{B}=(F,T)$. If the combined system is regarded as a system in and of itself, then it too should be associated with a test space. But how is this constructed, and how is it related to $\mathcal{A}$ and $\mathcal{B}$? Nothing that has been said with respect to single systems implies a unique answer, so further assumptions are needed.

Suppose that given separate systems $A$ and $B$, it is possible to perform any test $s$ on system $A$ simultaneously with any test $t$ on system $B$. A joint state assigns probabilities to pairs $(e,f)$ of outcomes. In particular, if $e\in s$ and $e\in s'$, then the probability of obtaining $(e,f)$ does not depend on whether the tests performed are $s$ and $t$, or $s'$ and $t$. Similarly if $f\in t$ and $f\in t'$. Suppose further that a specification of the joint probability for all outcome pairs serves to define the joint state uniquely.   

This motivates the following. 

\begin{Def}
Given two test spaces, $\mathcal{A}=(E,S)$ and $\mathcal{B}=(F,T)$, the \emph{Cartesian product}, $\mathcal{A}\times\mathcal{B}$, is a new test space whose set of outcomes is the set theoretic Cartesian product $E\times F$, and whose set of tests is $\{s\times t\,|s\in S,t\in T\}$, where $s\times t$ is again the set theoretic Cartesian product. 
\end{Def}

From Definition~\ref{states}, a state on $\mathcal{A}\times\mathcal{B}$ is a map $E\times F\rightarrow [0,1]$, with probabilities summing to $1$ for each pair of tests $(s,t)$.

\begin{Def}
Consider a bipartite system corresponding to a test space $\mathcal{A}\times \mathcal{B}$. A state $\omega\in\Omega(\mathcal{A}\times\mathcal{B})$ is \emph{nonsignalling} if and only if
\begin{align}
\sum_{f\in t}\ \omega(e,f) &= \sum_{f\in t'}\ \omega(e,f) \quad \forall e,t,t'\\
\sum_{e\in s}\ \omega(e,f) &= \sum_{e\in s'}\ \omega(e',f) \quad \forall f,s,s'.
\end{align}
\end{Def} 

If a state is nonsignalling, then the marginal probability of obtaining outcome $e$ for test $s$ does not depend on which $\mathcal{B}$ test is performed. This means that a marginal state $\omega_A\in\Omega(\mathcal{A})$, can be defined such that 
\begin{equation}
\omega_A(e) = \sum_{f\in t} \omega(e,f),
\end{equation}
where the right hand side does not depend on the choice of $t$. Similarly, one can define a marginal $\omega_B\in \Omega(\mathcal{B})$.

\begin{Def}
Given $\omega_A\in \Omega(\mathcal{A})$ and $\omega_B\in\Omega(\mathcal{B})$, the \emph{direct product}, $\omega_A\otimes\omega_B\in \Omega(\mathcal{A}\times\mathcal{B})$, is defined so that
\begin{equation}
(\omega_A\otimes\omega_B)(e,f) = \omega_A(e)\omega_B(f) \quad\forall (e,f)\in E\times F.
\end{equation}
\end{Def}

The definition of the Cartesian product of test spaces is valid for any pair of test spaces, including the case in which one of them is itself a Cartesian product. Considering three test spaces, $\mathcal{A}=(E,S)$, $\mathcal{B}=(F,T)$, and $\mathcal{C}=(G,U)$, it is easy to see that $(\mathcal{A}\times\mathcal{B})\times\mathcal{C}$ is isomorphic to $\mathcal{A}\times(\mathcal{B}\times\mathcal{C})$. Thus one can simply write $\mathcal{A}\times\mathcal{B}\times\mathcal{C}$, with states on $\mathcal{A}\times\mathcal{B}\times\mathcal{C}$ identified with maps $E\times F\times G \rightarrow [0,1]$. The product $\mathcal{A}\times\mathcal{A}\times\cdots\times\mathcal{A}$, where there are $n$ terms in the decomposition, can be written $\mathcal{A}^{\times n}$. 

The notion of a nonsignalling state has been defined with respect to bipartite decompositions. It extends readily to the case of an $n$-fold product. 
\begin{Def}\label{nonsigdef}
Consider a product of test spaces, $\mathcal{A}_1\times\cdots\times\mathcal{A}_n$. A state $\omega\in \Omega(\mathcal{A}_1\times\cdots\times\mathcal{A}_n)$ is \emph{n-fold nonsignalling} if and only if 
$\omega$ is nonsignalling with respect to every bipartite split. More formally, let $\alpha$ be a subset of $\{1,\ldots,n\}$, let $|\alpha|=k$, and write $\alpha=\{i_1,\ldots,i_k\}$. Then $\omega$ is n-fold nonsignalling iff
\begin{align*}
\sum_{e_{i_1}\in t_{i_1}}\cdots\sum_{e_{i_k}\in t_{i_k}} &\omega(e_1,\ldots,e_n) \\
&= \sum_{e_{i_1}\in t_{i_1}'}\cdots\sum_{e_{i_k}\in t_{i_k}'}\omega(e_1,\ldots,e_n),
\end{align*} 
for all $\alpha$, for all $e_j$ with $j\notin \alpha$, and for all tests $t_{i_1},\ldots,t_{i_k}$ and $t_{i_1}',\ldots,t_{i_k}'$.
\end{Def}

Finally, this will be useful:
\begin{Lemma}\label{tensorprodlemma}
Consider the test space $\mathcal{A}_1\times\cdots\times\mathcal{A}_n$. The direct product states, of the form $\omega_1\otimes\cdots\otimes\omega_n$, span a subspace of $V(\mathcal{A}_1\times\cdots\times\mathcal{A}_n)$, and the subspace can be identified with $V(\mathcal{A}_1)\otimes\cdots\otimes V(\mathcal{A}_n)$. If a joint state $\omega$ is n-fold nonsignalling, then $\omega \in V(\mathcal{A}_1)\otimes\cdots\otimes V(\mathcal{A}_n)$, i.e., $\omega$ can be written as a linear combination of direct products.
\end{Lemma}
\begin{proof}
Begin with the $n=2$ case. The tensor product $V(\mathcal{A}_1)\otimes V(\mathcal{A}_2)$ can be defined as the set of bilinear maps $V^*(\mathcal{A}_1)\times V^*(\mathcal{A}_2)\rightarrow \Re$. Any direct product $\omega_1\otimes\omega_2$ defines such a map via $(a,b)\rightarrow a(\omega_1)b(\omega_2)$, and it is straightforward that these span $V(\mathcal{A}_1)\otimes V(\mathcal{A}_2)$. Now consider a nonsignalling joint state $\omega$. The fact that $\omega$ is nonsignalling permits the definition of the marginal state $\omega_1$. Define the conditional state $\omega_{2|e}$ such that $\omega_{2|e}(f)$ is the probability of outcome $f$ on system $2$, given that outcome $e$ was obtained on system $1$. Thus $\omega(e,f) = \omega_1(e)\omega_{2|e}(f)$, for all $(e,f)\in E\times F$. Note that $\omega_{2|e}\in \Omega(\mathcal{A}_2)$. Now suppose that $f$ and $\{g_i\}$ are elements of $F$ such that, considered as elements of $V^*(\mathcal{A}_2)$, $f=\sum_i r_i g_i$. Thus $\omega_2(f) = \sum_i r_i \omega_2(g_i)$ for all $\omega_2\in \Omega(\mathcal{A}_2)$. Then $\omega(e,f) = \omega_1(e)\omega_{2|e}(f)=\sum_i r_i \omega_1(e)\omega_{2|e}(g_i)=\sum_i r_i\omega(e,g_i)$. A value $\omega(e,b)$ can now be defined for arbitrary $b\in V^*(\mathcal{A}_2)$ by linear extension. Similar reasoning concludes that $\omega(e,b)$ is linear in the first argument, hence can be extended to $\omega(a,b)$ for $a\in V^*(\mathcal{A}_1)$. So $\omega$ defines a bilinear map $V^*(\mathcal{A}_1)\times V^*(\mathcal{A}_2)\rightarrow \Re$ as required. The extension to general $n$ is straightforward. 
\end{proof}

Note that if we have two quantum test spaces $\mathcal{A} = (\mathcal{P}(H_A),\mathcal{M}(H_B)) $ and $\mathcal{B}=(\mathcal{P}(H_B),\mathcal{M}(H_B))$ then the space of nonsignalling states on $\mathcal{A} \times \mathcal{B}$ is larger than the state space of $\mathcal{AB} = (\mathcal{P}(H_A \otimes H_B),\mathcal{M}(H_A \otimes H_B))$.  To see this, recall that the nonsignalling state space of $\mathcal{A} \times \mathcal{B}$ only has to be positive for all possible choices of local measurements on $\mathcal{A}$ and $\mathcal{B}$, whereas the test space  $\mathcal{AB}$ includes joint measurements, such as the Bell measurement for example.  Thus, the criteria to be a state on $\mathcal{AB}$ are more restrictive than those for nonsignalling states on $\mathcal{A} \times \mathcal{B}$.  Indeed, if we take a state on $\mathcal{AB}$ and perform a positive, but not completely positive, map on system $\mathcal{A}$ then the result is still a valid state on $\mathcal{A} \times \mathcal{B}$, but not on $\mathcal{AB}$ in general, e.g. consider performing a partial transpose on a Bell state.  Nevertheless, the state space of $\mathcal{AB}$ is still a convex subset of the nonsignalling states on $\mathcal{A} \times \mathcal{B}$, which is enough to apply our theorem.  More generally,  one might want to consider rules for composing subsystems that yield a convex subset of the nonsignalling states on the Cartesian product for arbitrary test spaces.

\section{A de Finetti theorem for test spaces}
\label{MainRes}

Given a system associated with a test space $\mathcal{A}=(E,S)$, suppose that $n$ copies are associated with the product $\mathcal{A}^{\times n}$. Given an infinite sequence of states $\omega^1,\omega^2,\ldots$ where $\omega^n\in\Omega(\mathcal{A}^{\times n})$, it is possible to define symmetry and exchangeability in a manner similar to the classical and quantum cases. The main difference is that here, the definition of exchangeability involves the extra condition that the states are nonsignalling.

\begin{Def}\label{symmetricdef}
A state $\omega^n\in\Omega(\mathcal{A}^{\times n})$ is \emph{symmetric} if and only if it is invariant under permutations of the $n$ systems. That is, 
\[
\omega^n(e_1,\ldots,e_n) = \omega^n(e_{\pi(1)},\ldots,e_{\pi(n)}),
\]
for all permutations $\pi:(1,\ldots,n)\rightarrow (1,\ldots,n)$.
\end{Def}

\begin{Def}\label{exchangedef}
A sequence of states $\omega^1,\omega^2,\ldots$ where $\omega^n\in\Omega(\mathcal{A}^{\times n})$, is \emph{exchangeable} if and only if
\begin{enumerate}
\item $\forall n$ $\omega^n$ is symmetric,
\item $\forall n$ $\omega^n$ is n-fold nonsignalling,\
\item $\omega^n(e_1,\ldots,e_n) = \sum_{e_{n+1}\in s} \,\omega^{n+1}(e_1,\ldots,e_n,e_{n+1})$.
\end{enumerate}
\end{Def}

\begin{Theorem}[The de Finetti theorem for test spaces]\label{testdefinetti}
Suppose that the sequence $\omega^1,\omega^2,\ldots$ where $\omega^n\in\Omega(\mathcal{A}^{\times n})$, is exchangeable. Then $\omega^n$ can be written in the form
\begin{equation}
\label{MainRes:DFF}
\omega^n = \int_{\Omega(\mathcal{A})} \!\mathrm{d}\mu(\omega) \ \omega\otimes\cdots\otimes\omega,
\end{equation}
where $\mu$ is a probability measure on $\Omega(\mathcal{A})$, $\mu$ is independent of $n$, and $\mu$ is unique.
\end{Theorem}
\begin{proof}
The proof is an adaptation of the proof of the quantum de Finetti theorem due to Caves et. al. \cite{Caves:2002aa}. Recall that in quantum theory, an informationally complete measurement is a positive operator-valued (POV) measurement such that if the outcome probabilities are all known, then the state is determined uniquely. The strategy of Caves et al. is to generate a classical distribution by considering an informationally complete POV measurement performed separately on each quantum system. Applying the classical de Finetti theorem to the distribution of outcome sequences allows the form of the quantum state to be inferred.

In our context, the test space $\mathcal{A}$ need not include a test that is informationally complete for the state space $\Omega$. But for the purposes of proof, this does not matter. All that is needed is a corresponding mathematical construction. 

\begin{Lemma}
\label{ICLem}
There exists a set $M = \{a_1,\ldots,a_d\}$, where $a_i\in V^*(\mathcal{A})$ and $d$ is the dimension of $V^*(\mathcal{A})$, such that 
\begin{enumerate}
\item $M$ is a basis for $V^*(\mathcal{A})$, i.e., the $a_i$ are linearly independent,
\item $0\leq a_i(\omega) \leq 1 \ \forall \omega\in\Omega(\mathcal{A})$,
\item $\sum_i a_i(\omega)=1 \ \forall \omega\in\Omega(\mathcal{A})$.
\end{enumerate}
\end{Lemma}
\begin{proof}
This result is not new. It is also used in Ref.~\cite{Barnum:2006}, and we give the same proof. A more general version is proven in Ref.~\cite{Singer:1992}. Let $u$ be the unique vector in $V^*(\mathcal{A})$ such that $u(\omega)=1$ for all $\omega\in\Omega(\mathcal{A})$. Consider an arbitrary basis $\{b_1,\ldots,b_d\}$ of $V^*(\mathcal{A})$. Apply an invertible linear transformation to obtain $\{\tilde{b}_1,\ldots,\tilde{b}_d\}$, where $\sum_i \tilde{b}_i = u$. Some of the $\tilde{b}_i$ may be negative on some states $\omega$. Define a constant $c$ as the minimum value of $\tilde{b}_i(\omega)$, where the minimum is taken over all $i$ and all $\omega\in\Omega(\mathcal{A})$. Then $\tilde{b}_i - cu$ is positive and non-zero. It follows that $\sum_i (\tilde{b}_i-cu)=(1-dc)u$, with $1-dc>0$. Define $a_i = (\tilde{b}_i-cu)/(1-dc)$.  
\end{proof}

The set $M$ will play the role of an informationally complete measurement. The linear independence of the $a_i$ means that a state $\omega$ is determined uniquely by the values $a_i(\omega)$. The idea now is that given a nonsignalling $\omega^n\in \Omega(\mathcal{A}^{\times n})$, one can at least imagine a measurement corresponding to $M$ performed separately on each system. The probability of obtaining an outcome sequence $(a_{i_1},\ldots,a_{i_n})$ for a direct product state $\omega_1\otimes\cdots\otimes\omega_n$ is defined as $a_{i_1}(\omega_1)\cdots a_{i_n}(\omega_n)$. Recalling Lemma~\ref{tensorprodlemma}, according to which a nonsignalling state can be written as a linear combination of direct product states, the probability of the sequence $(a_{i_1},\ldots,a_{i_n})$ for an arbitrary nonsignalling state can be defined by linear extension. Note that this way, the sequence $(a_{i_1},\ldots,a_{i_n})$ corresponds to a vector $a_{i_1}\otimes\cdots\otimes a_{i_n}\in \left(V^*(\mathcal{A})\right)^{\otimes n}$, such that the probability is given by $(a_{i_1}\otimes\cdots\otimes a_{i_n})(\omega^n)$. The vectors $a_{i_1}\otimes\cdots\otimes a_{i_n}$ are linearly independent and span the tensor product space $\left(V^*(\mathcal{A})\right)^{\otimes n}$. This means that the joint measurement $M^{\times n}$ is informationally complete for the nonsignalling $n$-partite system.

For simplicity, write the probability of the outcome sequence $(a_{i_1},\ldots,a_{i_n})$ as $P^n(i_1,\ldots,i_n)$ so that
\begin{equation}
P^n(i_1,\ldots,i_n) = (a_{i_1}\otimes\cdots\otimes a_{i_n})(\omega^n).
\end{equation}
If $\omega^n$ is symmetric, then so is $P^n$, since
\begin{align}
P^n(i_{\pi(1)},\ldots,i_{\pi(n)}) &= (a_{i_{\pi(1)}}\otimes\cdots\otimes a_{i_{\pi(n)}})(\omega^n) \nonumber \\
&= (a_{i_1}\otimes\cdots\otimes a_{i_n})(\Pi(\omega^n)) \nonumber \\
&= (a_{i_1}\otimes\cdots\otimes a_{i_n})(\omega^n) \nonumber \\
&= P^n(i_1,\ldots,i_n),
\end{align}
for any permutation $\pi$.

Further, if the sequence $\omega^1,\omega^2,\ldots$ is exchangeable, then so is the sequence $P^1,P^2,\ldots$, since it is symmetric and
\begin{align}
\sum_{i_{n+1}}P^{n+1}&(i_1,\ldots,i_n,i_{n+1}) \nonumber\\
&= \sum_{i_{n+1}} (a_{i_1}\otimes\cdots\otimes a_{i_n}\otimes a_{i_{n+1}})(\omega^{n+1})  \nonumber \\
&= (a_{i_1}\otimes\cdots\otimes a_{i_n})(\omega^n) \nonumber \\
&= P^n(i_1,\ldots,i_n).
\end{align}
Now apply the classical de Finetti theorem (Theorem \ref{classicaldefinetti}) to the sequence $P^1,P^2,\ldots$ to obtain
\begin{equation}\label{classexpression}
P^n(i_1,\ldots,i_n) = \int_{\Delta_d} \!\mathrm{d}\mu(p) \ p(a_{i_1}) \cdots p(a_{i_n}),
\end{equation}
where $\Delta_d$ is the set of probability distributions over $\{a_1,\ldots,a_d\}$. 

For each such distribution $p$, there is a unique $\omega_p\in V(\mathcal{A})$ such that $a_i(\omega_p)=p(a_i)$ (where uniqueness follows from the linear independence of the $a_i$). Thus Eq.(\ref{classexpression}) can be rewritten
\begin{equation}
(a_{i_1}\otimes\cdots\otimes a_{i_n})(\omega^n) = \int_{\Delta_d} \!\mathrm{d}\mu(p) \ a_{i_1}(\omega_p)\cdots a_{i_n}(\omega_p).
\end{equation}
Since the joint measurement is informationally complete for the nonsignalling $n$-partite states, this implies
\begin{equation}\label{notquitethere}
\omega^n = \int_{\Delta_d} \!\mathrm{d}\mu(p) \ \omega_p\otimes\cdots\otimes\omega_p.
\end{equation}

This is not quite sufficient to establish Theorem~\ref{testdefinetti}. The right hand side of Eq.(\ref{notquitethere}) is an integral over all $\omega_p\in V(\mathcal{A})$ satisfying $0\leq a_i(\omega_p)\leq 1$ and $\sum_i a_i(\omega_p)=1$. That is, it is an integral over all $\omega_p$ that return sensible probabilities for the imaginary informationally complete measurement. But in general, there are $\omega_p$ satisfying these conditions that are not valid states because they return a value $<0$ for some element of the test space. It remains to show that the integral can be restricted to those $\omega_p\in\Omega(\mathcal{A})$. 

To this end, consider an $\omega_p\in V(\mathcal{A})$ such that $0\leq a_i(\omega_p)\leq 1$ and $\sum_i a_i(\omega_p)=1$, but $\omega_p(e)<0$ for some $e\in E$. It must be the case that $e\in s$ for some test $s$. Let the other elements of $s$ be $\{f_1,\ldots,f_k\}$ and note that $\sum_i \omega_p(f_i)>1$. There must exist an $\epsilon>0$ and a neighborhood $N$ of $\omega_p$ in $V(\mathcal{A})$ such that $\sum_i\omega(f_i)>1+\epsilon$ for all $\omega\in N$. Let $\tilde{N}$ be the subset of $\Delta_d$ such that $p(a_i)=a_i(\omega)$ for some $\omega\in N$. The expression (\ref{notquitethere}) holds for any $n$, so suppose that $n$ is even and that the test $s$ is performed on each of the $n$ systems. The probability that outcome $e$ is never obtained is defined by $\omega^n$ and is given by
\begin{align}
& \int_{\Delta_d} \!\mathrm{d}\mu(p) \ \left(\sum_i \omega_p(f_i) \right)^n \nonumber \\
& = \int_{\Delta_d \backslash \tilde{N}} \!\mathrm{d}\mu(p)\ \left(\sum_i \omega_p(f_i) \right)^n + \int_{\tilde{N}} \!\mathrm{d}\mu(p)\ \left(\sum_i \omega_p(f_i) \right)^n \nonumber \\
& \geq \int_{\tilde{N}} \!\mathrm{d}\mu(p)\ \left(\sum_i \omega_p(f_i) \right)^n \nonumber \\
& \geq (1+\epsilon)^n \int_{\tilde{N}}\!\mathrm{d}\mu(p),
\end{align}
where we have used the fact that the first term in the second line is $\geq 0$ if $n$ is even.     
For large enough $n$ this expression is $>1$ unless $\mu(\tilde{N})=0$. This holds for any $\omega_p\notin \Omega(\mathcal{A})$. It follows that, with a suitable redefinition of $\mu$, 
\begin{equation}
\omega^n = \int_{\Omega(\mathcal{A})} \!\mathrm{d}\mu(\omega) \ \omega\otimes\cdots\otimes\omega.
\end{equation}
\end{proof}

\section{Consequences of Theorem \ref{testdefinetti}}

\label{Cons}

The de Finetti theorem for test spaces is rather general. Very little goes into the definition of a test space itself.  In fact, in the finite dimensional case we are considering, the state space of an individual system may be an arbitrary compact convex set.  Thus the most substantive assumptions that go into the theorem are those that concern how test spaces combine when joint systems are considered.  

In the classical case, the test space $(E,\{E\})$ contains a single test and the state space for an individual system is $\Delta_d$, where $d$ is the number of elements of $E$.  In this case, the nonsignalling condition is redundant because there is just a single test on each system and hence no freedom to choose alternative measurements.  The Cartesian product of classical test spaces corresponds to the usual Cartesian product of sample spaces and so Theorem~\ref{testdefinetti} reduces to Theorem~\ref{classicaldefinetti} straightforwardly. The quantum case is a little more subtle because the state of $n$ quantum systems belongs to the state space of $(\mathcal{P}(H^{\otimes n}),\mathcal{M}(H^{\otimes n}))$, rather than the state space of the Cartesian product of $n$ test spaces of the form $(\mathcal{P}(H),\mathcal{M}(H))$. Nevertheless, a quantum state defined by a density operator on $H^{\otimes n}$ is uniquely specified by the probabilities for measurement outcomes of the form $Q_1\otimes Q_2\otimes\cdots\otimes Q_n$, where the $Q_i$ are projection operators. It follows that the states on $(\mathcal{P}(H^{\otimes n}),\mathcal{M}(H^{\otimes n}))$ can be identified with a subset of the nonsignalling states on $(\mathcal{P}(H),\mathcal{M}(H))^{\times n}$. Hence Theorem~\ref{testdefinetti} implies the quantum de Finetti theorem. 

As a further illustration of the generality of the result, we note that a theorem for classical processes, or conditional probabilities, can also be viewed as a special case of Theorem~\ref{testdefinetti}. A process $C$ can be thought of as taking an input $Y$ into an output $X$, where $Y$ takes values in $\{1,\ldots,k\}$ and $X$ takes values in $\{1,\ldots, d\}$. The process can be defined as the set of conditional probabilities of the form $P(X = x|Y = y)$ (abbreviated $P(x|y)$). Clearly, the set of all such $C$ can be regarded as the set of states on a test space $(E,S)$, where $E$ consists of ordered pairs with $E = \{(x,y)\}_{x=1,\ldots, d, y = 1,\ldots, k}$ and $S = \{\{(x,1)\}_{x=1}^d,\dots,\{(x,k)\}_{x=1}^d\}$, as illustrated for the $d = k = 3$ case in Fig.~\ref{BoxGreechie}.  
\begin{figure}
\includegraphics[scale=0.7]{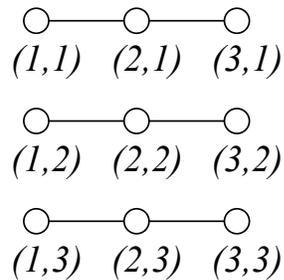}
\caption{The test space for a process with three inputs and three outputs.\label{BoxGreechie}}
\end{figure}
The state space of this test space is isomorphic to the set of conditional probability distributions via the identification $\omega((x,y)) = P(x|y)$.

More generally, a process $C^n$ takes inputs $Y_1,\ldots,Y_n$ into outputs $X_1,\ldots,X_n$, where each $Y_i$ takes values in $\{1,\ldots,k\}$, and each $X_i$ takes values in $\{1,\ldots, d\}$. Such a $C^n$ can be defined as the set of conditional probabilities of the form $P^n(x_1,\ldots,x_n|y_1,\ldots,y_n)$, and can also be identified with a state on the Cartesian product $(E,S)^{\times n}$. The n-fold nonsignalling and symmetry conditions can be defined for $C^n$ exactly as they are in Definitions~\ref{nonsigdef} and \ref{symmetricdef}. Exchangeability for a sequence $C^1,C^2,\ldots$ can be defined exactly as in Definition~\ref{exchangedef}. Theorem~\ref{testdefinetti} becomes the following theorem for classical processes.
\begin{Theorem}\label{conddefinetti}
If the sequence $C^1,C^2,\ldots$ is exchangeable, then the conditional probabilities defining $C^n$ can be written in the form
\begin{equation}\label{condindep}
P(x_1,\ldots,x_n|y_1,\ldots,y_n) = \int_{\Omega}\!\mathrm{d}\mu(p) \ p(x_1|y_1)\cdots p(x_n|y_n),
\end{equation}
where $\Omega$ is the set of processes that take a single input $X$ into a single output $Y$, $\mu$ is a probability measure on $\Omega$, $\mu$ is independent of $n$, and $\mu$ is unique. 
\end{Theorem}

A couple of remarks regarding Theorem~\ref{conddefinetti} might be helpful.

First, we described Theorem~\ref{conddefinetti} as a special case of Theorem~\ref{testdefinetti}. It is worth noting the sense in which it is strictly less general. After all, any state on any test space could be thought of as a classical process, which takes an input (choice of test) into an output (outcome of test). The point is that, with this identification applied to a generic test space $\mathcal{A}$, not all sets of conditional probabilities $p(x|y)$ will correspond to valid states on $\mathcal{A}$. It is only if $\mathcal{A}$ has the special feature that the tests are non-overlapping that this will be the case. The integral in Eq.(\ref{condindep}) ranges over all processes of the form $p(x|y)$, whereas in Eq.(\ref{MainRes:DFF}), it is important that the integral ranges only over $\Omega(\mathcal{A})$. 

Second, the de Finetti theorem for classical processes can be viewed as a de Finetti theorem for \emph{states} in a theory that is non-classical and non-quantum, but admits the most general type of correlations compatible with the nonsignalling requirement.  This theory is discussed in \cite{Barrett:2005aa}, where it is called Generalized Non-signalling Theory. It admits superquantum correlations, which have been discussed in the quantum information literature under the name Popescu-Rohrlich, or nonlocal, boxes.  Of course, as in the quantum case, exchangeable states do not actually exhibit such correlations, since the de Finetti theorem shows that they are separable.

\section{When does a de Finetti-type theorem not hold?}
\label{Not}

The de Finetti theorem for test spaces holds thanks to a number of assumptions concerning how systems combine to make joint systems. One of these is the nonsignalling condition for joint states. Others are encoded in the formal definition of a Cartesian product of test spaces. It is interesting to see what happens when these assumptions are relaxed, so in this section we present a number of cases where the theorem fails.

\subsection{The nonsignalling condition}
\label{Not:Signal}

First, the assumption that the joint states are nonsignalling is crucial not only in the proof, but in the very definition of exchangeability. In general, if a state $\omega\in\Omega(\mathcal{A}_1\times\mathcal{A}_2)$ is signalling, then it is not possible to define marginal states $\omega_1$ and $\omega_2$. But exchangeability requires that, given the $n+1$th state in the sequence, the marginal state of the first $n$ systems should be defined and should equal the $n$th state. In fact, if a state $\omega\in\Omega(\mathcal{A}_1\times\mathcal{A}_2)$ is only signalling in one direction then one of the marginals can be defined. For example, if probabilities of outcomes for system 2 depend on which test was performed on system 1, but not vice versa, then $\omega_1$ is well defined. But such a state is not symmetric, thus could not form part of an exchangeable sequence. Arguably, the possibility of performing tests on one system that do not affect the other is part and parcel of what we mean when we speak of separate systems (or separate trials).

\subsection{Simultaneous measurements}
\label{Not:Sim}

Implicit in the definition of the Cartesian product of test spaces is the idea that a test on one system can be regarded as simultaneous with a test on the other system. One can certainly imagine rules for combining systems where this is not the case. As a very simple example, consider two classical bits which have combined in the following strange manner. If bit 1 is measured before bit 2, then the bits are found to be $00$ or $11$ with equal probability. On the other hand, if bit 2 is measured before bit 1, then the outcomes are $01$ and $10$ with equal probability. Note that a suitable no-signalling condition is satisfied and it is possible to define marginal states for these bits. With more complicated test spaces one can construct examples like this which are also symmetric.\footnote{Consider the test space $\mathcal{A}$, with outcomes $E=\{a,b,c,d\}$, and two tests corresponding to $s1=\{a,b\}$, and $s2=\{c,d\}$. Suppose that there are two systems, $A$ and $B$, each described by the test space $\mathcal{A}$, and that they have combined as follows. If the test $s1$ is performed on both systems, the outcomes are completely random and uncorrelated. If the test $s2$ is performed on both systems, then the outcomes are completely random and uncorrelated. On the other hand, if the test $s1$ is performed on either one of the systems, followed by $s2$ on the other, then the joint outcomes are $ac$ or $bd$ with equal probability. If the test $s2$ is performed on either one of the systems, followed by $s1$ on the other, then the joint outcomes are $ad$ or $bc$ with equal probability. It is clear that this peculiar bipartite system does not allow signalling and is also invariant under a permutation of the two systems.}  We leave open the status of the de Finetti theorem in such cases.  

\subsection{Extra degrees of freedom}
\label{Not:Extra}

Another assumption that is implicit in the Cartesian product of test spaces is that the joint state of two systems is completely specified by the probabilities for the joint outcomes $(e,f)$ of each pair of local tests $(s,t)$. 

One can construct theories in which a joint state does indeed determine such probabilities, but is not completely specified by them. There are extra degrees of freedom, bound up in the two systems, which are inaccessible unless some kind of joint operation involving both systems at once is performed. 

As discussed in \cite{Caves:2002aa}, a clear example of this is provided by a modification of quantum theory in which real Hilbert spaces are used rather than complex Hilbert spaces. States and observables correspond to real symmetric, rather than complex Hermitian, operators. For a $2$-dimensional system (a \emph{rebit}), there are measurements corresponding to $x$- and $z$-spin but not $y$-spin. In this case, if $\sigma_y$ is the usual Pauli matrix, then $1/4 (I\otimes I + \sigma_y \otimes \sigma_y )$ is an allowed state of two rebits. But the only way it can be distinguished from $1/4(I\otimes I)$ is if a joint observable such as $\sigma_y\otimes\sigma_y$ is measured. Note that in real quantum theory, this really is a joint observable: it cannot be measured via a separate $\sigma_y$ measurement on each system.

In \cite{Caves:2002aa} it is shown explicitly that the de Finetti theorem fails in real quantum theory. If $\omega^n$ is defined by
\begin{equation}
\label{Not:Counter}
\omega^n = \frac12 \left(\frac{I+\sigma_y}{2}\right)^{\otimes n} + \frac12 \left(\frac{I-\sigma_y}{2}\right)^{\otimes n}
\end{equation}
then it is real and symmetric, and the sequence $\omega^1,\omega^2,\ldots$ is exchangeable. But by the de Finetti theorem for complex quantum theory, the right hand side of Eq.~\eqref{Not:Counter} is the unique de Finetti representation for this sequence.  

\section{Conclusion}
\label{Conc}

In this paper, an infinite de Finetti theorem for test spaces has been presented, which generalizes both the classical and quantum de Finetti theorems.  To illustrate the generality of the result, we have shown that a de Finetti theorem for classical processes, which may also be interpreted as a de Finetti theorem for nonlocal boxes, follows as a special case.

From a practical point of view, proving theorems for test spaces, rather than just for quantum theory, confers significant advantages.  Not only do we achieve a unification of the classical and quantum results, but we also obtain results that apply to essentially arbitrary convex sets. This is potentially relevant when technological limitations prevent the preparation of arbitrary quantum states of certain systems, so that there is an effective restriction to a convex subset.   

From a foundational point of view, this work can be seen as part of a project of understanding what is responsible for the enhanced information processing power of quantum theory, and for the project of deriving quantum theory from information theoretic axioms.  In particular, if one adopts a subjective Bayesian approach to probability, it might be desirable to impose the requirement that, in any reasonable theory, one should be able to make sense of the idea of reconstructing an unknown state of a system by making repeated measurements.  Having a de Finetti theorem for test spaces means that this does indeed make sense for theories in this framework, and that the existing approaches to Bayesian state tomography in quantum theory would generalize straightforwardly.  

There are various directions for future work. It would be useful to produce a finite de Finetti theorem for test spaces. It would also be useful to establish whether a (finite or infinite) de Finetti theorem holds without the assumption made here of finite dimensionality of state spaces. Finally, as discussed in Section~\ref{Not:Sim}, it might be interesting to explore the status of de Finetti-type theorems in cases where systems combine in non-standard ways.

\emph{Note added.} Related results have recently been obtained by M.~Christandl and B.~Toner \cite{Christandl:2007}, who derive a de Finetti theorem for classical processes, analogous to Theorem~\ref{conddefinetti} of the present work, but extended to the finite case.

\begin{acknowledgments}
We would like to thank Carl Caves, Matthias Christandl, Renato Renner and Ben Toner for useful discussions about the de Finetti theorem.  Research at Perimeter Institute for Theoretical Physics is supported in part by the Government of Canada through NSERC and by the Province of Ontario through MRI. At IQC, ML was supported in part by MITACS and ORDCF. ML was also supported in part by grant RFP1-06-006 from The Foundational Questions Institute (fqxi.org). Part of this work was carried out while JB was supported by an HP research fellowship and by the EU FP6-FET Integrated Project SCALA (CT-015714). JB is supported by an EPSRC Career Acceleration Fellowship.
\end{acknowledgments}


\begin{thebibliography}{55}
\expandafter\ifx\csname natexlab\endcsname\relax\def\natexlab#1{#1}\fi
\expandafter\ifx\csname bibnamefont\endcsname\relax
  \def\bibnamefont#1{#1}\fi
\expandafter\ifx\csname bibfnamefont\endcsname\relax
  \def\bibfnamefont#1{#1}\fi
\expandafter\ifx\csname citenamefont\endcsname\relax
  \def\citenamefont#1{#1}\fi
\expandafter\ifx\csname url\endcsname\relax
  \def\url#1{\texttt{#1}}\fi
\expandafter\ifx\csname urlprefix\endcsname\relax\def\urlprefix{URL }\fi
\providecommand{\bibinfo}[2]{#2}
\providecommand{\eprint}[2][]{\url{#2}}

\bibitem[{\citenamefont{Renner}(2005)}]{Renner:2005aa}
\bibinfo{author}{\bibfnamefont{R.}~\bibnamefont{Renner}}, Ph.D. thesis,
  \bibinfo{school}{Swiss Federal Institute of Technology}
  (\bibinfo{year}{2005}), \eprint{quant-ph/0512258}.

\bibitem[{\citenamefont{Doherty et~al.}(2004)\citenamefont{Doherty, Parrilo,
  and Spedalieri}}]{Doherty:2004}
\bibinfo{author}{\bibfnamefont{A.~C.} \bibnamefont{Doherty}},
  \bibinfo{author}{\bibfnamefont{P.~A.} \bibnamefont{Parrilo}},
  \bibnamefont{and} \bibinfo{author}{\bibfnamefont{F.~M.}
  \bibnamefont{Spedalieri}}, \bibinfo{journal}{Phys. Rev. A}
  \textbf{\bibinfo{volume}{69}}, \bibinfo{pages}{022308}
  (\bibinfo{year}{2004}), \eprint{quant-ph/0308032}.

\bibitem[{\citenamefont{Ioannou}(2007)}]{Ioannou:2007}
\bibinfo{author}{\bibfnamefont{L.~M.} \bibnamefont{Ioannou}},
  \bibinfo{journal}{Quantum Information and Computation}
  \textbf{\bibinfo{volume}{7}}, \bibinfo{pages}{335} (\bibinfo{year}{2007}),
  \eprint{quant-ph/0603199}.

\bibitem[{\citenamefont{Audenaert}(2004)}]{Audenaert:2004}
\bibinfo{author}{\bibfnamefont{K.~M.~R.} \bibnamefont{Audenaert}}, in
  \emph{\bibinfo{booktitle}{Proceedings of MTNS2004}}, edited by
  \bibinfo{editor}{\bibfnamefont{B.}~\bibnamefont{De~Moor}},
  \bibinfo{editor}{\bibfnamefont{B.}~\bibnamefont{Motmans}},
  \bibinfo{editor}{\bibfnamefont{J.}~\bibnamefont{Wilems}},
  \bibinfo{editor}{\bibfnamefont{P.}~\bibnamefont{Van~Dooren}},
  \bibnamefont{and} \bibinfo{editor}{\bibfnamefont{V.}~\bibnamefont{Blondel}}
  (\bibinfo{year}{2004}), \eprint{quant-ph/0402076}.

\bibitem[{\citenamefont{Barrett et~al.}(2005)\citenamefont{Barrett, Linden,
  Massar, Pironio, Popescu, and Roberts}}]{Barrett:2005a}
\bibinfo{author}{\bibfnamefont{J.}~\bibnamefont{Barrett}},
  \bibinfo{author}{\bibfnamefont{N.}~\bibnamefont{Linden}},
  \bibinfo{author}{\bibfnamefont{S.}~\bibnamefont{Massar}},
  \bibinfo{author}{\bibfnamefont{S.}~\bibnamefont{Pironio}},
  \bibinfo{author}{\bibfnamefont{S.}~\bibnamefont{Popescu}}, \bibnamefont{and}
  \bibinfo{author}{\bibfnamefont{D.}~\bibnamefont{Roberts}},
  \bibinfo{journal}{Phys. Rev. A} \textbf{\bibinfo{volume}{71}},
  \bibinfo{pages}{022101} (\bibinfo{year}{2005}).

\bibitem[{\citenamefont{Barrett}(2007)}]{Barrett:2005aa}
\bibinfo{author}{\bibfnamefont{J.}~\bibnamefont{Barrett}},
  \bibinfo{journal}{Phys. Rev. A} \textbf{\bibinfo{volume}{75}},
  \bibinfo{pages}{032304} (\bibinfo{year}{2007}), \eprint{quant-ph/0508211}.

\bibitem[{\citenamefont{Barrett and Pironio}(2005)}]{Barrett:2005b}
\bibinfo{author}{\bibfnamefont{J.}~\bibnamefont{Barrett}} \bibnamefont{and}
  \bibinfo{author}{\bibfnamefont{S.}~\bibnamefont{Pironio}},
  \bibinfo{journal}{Phys. Rev. Lett.} \textbf{\bibinfo{volume}{95}},
  \bibinfo{pages}{140401} (\bibinfo{year}{2005}).

\bibitem[{\citenamefont{Brassard et~al.}(2006)\citenamefont{Brassard, Buhrman,
  Linden, M{\'e}thot, Tapp, and Unger}}]{Brassard:2005}
\bibinfo{author}{\bibfnamefont{G.}~\bibnamefont{Brassard}},
  \bibinfo{author}{\bibfnamefont{H.}~\bibnamefont{Buhrman}},
  \bibinfo{author}{\bibfnamefont{N.}~\bibnamefont{Linden}},
  \bibinfo{author}{\bibfnamefont{A.~A.} \bibnamefont{M{\'e}thot}},
  \bibinfo{author}{\bibfnamefont{A.}~\bibnamefont{Tapp}}, \bibnamefont{and}
  \bibinfo{author}{\bibfnamefont{F.}~\bibnamefont{Unger}},
  \bibinfo{journal}{Phys. Rev. Lett.} \textbf{\bibinfo{volume}{96}},
  \bibinfo{pages}{250401} (\bibinfo{year}{2006}), \bibinfo{note}{arXiv.org
  e-print {\tt quant-ph/0508042}}.

\bibitem[{\citenamefont{Broadbent and M{\'e}thot}(2006)}]{Broadbent:2005}
\bibinfo{author}{\bibfnamefont{A.}~\bibnamefont{Broadbent}} \bibnamefont{and}
  \bibinfo{author}{\bibfnamefont{A.~A.} \bibnamefont{M{\'e}thot}},
  \bibinfo{journal}{Theor. Comput. Sci.} \textbf{\bibinfo{volume}{358}},
  \bibinfo{pages}{3} (\bibinfo{year}{2006}), \bibinfo{note}{arXiv.org e-print
  {\tt quant-ph/0504136}}.

\bibitem[{\citenamefont{Buhrman et~al.}(2006)\citenamefont{Buhrman, Christandl,
  Unger, Wehner, and Winter}}]{Buhrman:2005}
\bibinfo{author}{\bibfnamefont{H.}~\bibnamefont{Buhrman}},
  \bibinfo{author}{\bibfnamefont{M.}~\bibnamefont{Christandl}},
  \bibinfo{author}{\bibfnamefont{F.}~\bibnamefont{Unger}},
  \bibinfo{author}{\bibfnamefont{S.}~\bibnamefont{Wehner}}, \bibnamefont{and}
  \bibinfo{author}{\bibfnamefont{A.}~\bibnamefont{Winter}},
  \bibinfo{journal}{Proc. Royal Soc. A} \textbf{\bibinfo{volume}{462}},
  \bibinfo{pages}{1919} (\bibinfo{year}{2006}), \bibinfo{note}{arXiv.org
  e-print {\tt quant-ph/0504133}}.

\bibitem[{\citenamefont{van Dam}(2000)}]{Dam:2000}
\bibinfo{author}{\bibfnamefont{W.}~\bibnamefont{van Dam}}, Ph.D. thesis,
  \bibinfo{school}{University of Oxford} (\bibinfo{year}{2000}).

\bibitem[{\citenamefont{van Dam}(2005)}]{Dam:2005}
\bibinfo{author}{\bibfnamefont{W.}~\bibnamefont{van Dam}}
  (\bibinfo{year}{2005}), \bibinfo{note}{arXiv.org e-print {\tt
  quant-ph/0501159}}.

\bibitem[{\citenamefont{Jones and Masanes}(2005)}]{Jones:2005}
\bibinfo{author}{\bibfnamefont{N.~S.} \bibnamefont{Jones}} \bibnamefont{and}
  \bibinfo{author}{\bibfnamefont{L.}~\bibnamefont{Masanes}},
  \bibinfo{journal}{Phys. Rev. A} \textbf{\bibinfo{volume}{72}},
  \bibinfo{pages}{052312} (\bibinfo{year}{2005}).

\bibitem[{\citenamefont{Khalfi and Tsirelson}(1985)}]{Khalfi:1985}
\bibinfo{author}{\bibfnamefont{L.~A.} \bibnamefont{Khalfi}} \bibnamefont{and}
  \bibinfo{author}{\bibfnamefont{B.~S.} \bibnamefont{Tsirelson}}, in
  \emph{\bibinfo{booktitle}{Symposium on the Foundations of Modern Physics}},
  edited by \bibinfo{editor}{\bibfnamefont{P.}~\bibnamefont{Lahti}}
  \bibnamefont{and}
  \bibinfo{editor}{\bibfnamefont{P.}~\bibnamefont{Mittelstaedt}}
  (\bibinfo{publisher}{World Scientific, Singapore}, \bibinfo{year}{1985}), pp.
  \bibinfo{pages}{441--460}.

\bibitem[{\citenamefont{Popescu and Rohrlich}(1994)}]{Popescu:1994}
\bibinfo{author}{\bibfnamefont{S.}~\bibnamefont{Popescu}} \bibnamefont{and}
  \bibinfo{author}{\bibfnamefont{D.}~\bibnamefont{Rohrlich}},
  \bibinfo{journal}{Found. Phys.} \textbf{\bibinfo{volume}{24}},
  \bibinfo{pages}{379} (\bibinfo{year}{1994}).

\bibitem[{\citenamefont{Short et~al.}(2006{\natexlab{a}})\citenamefont{Short,
  Gisin, and Popescu}}]{Short:2005}
\bibinfo{author}{\bibfnamefont{A.}~\bibnamefont{Short}},
  \bibinfo{author}{\bibfnamefont{N.}~\bibnamefont{Gisin}}, \bibnamefont{and}
  \bibinfo{author}{\bibfnamefont{S.}~\bibnamefont{Popescu}},
  \bibinfo{journal}{Quantum Information Processing}
  \textbf{\bibinfo{volume}{5}}, \bibinfo{pages}{131}
  (\bibinfo{year}{2006}{\natexlab{a}}), \bibinfo{note}{arXiv.org e-print {\tt
  quant-ph/0504134}}.

\bibitem[{\citenamefont{Short et~al.}(2006{\natexlab{b}})\citenamefont{Short,
  Popescu, and Gisin}}]{Short:2005a}
\bibinfo{author}{\bibfnamefont{A.}~\bibnamefont{Short}},
  \bibinfo{author}{\bibfnamefont{S.}~\bibnamefont{Popescu}}, \bibnamefont{and}
  \bibinfo{author}{\bibfnamefont{N.}~\bibnamefont{Gisin}},
  \bibinfo{journal}{Phys. Rev. A} \textbf{\bibinfo{volume}{73}},
  \bibinfo{pages}{012101} (\bibinfo{year}{2006}{\natexlab{b}}),
  \bibinfo{note}{arXiv.org eprint {\tt quant-ph/0508120}}.

\bibitem[{\citenamefont{Barnum et~al.}(2006)\citenamefont{Barnum, Barrett,
  Leifer, and Wilce}}]{Barnum:2006}
\bibinfo{author}{\bibfnamefont{H.}~\bibnamefont{Barnum}},
  \bibinfo{author}{\bibfnamefont{J.}~\bibnamefont{Barrett}},
  \bibinfo{author}{\bibfnamefont{M.}~\bibnamefont{Leifer}}, \bibnamefont{and}
  \bibinfo{author}{\bibfnamefont{A.}~\bibnamefont{Wilce}}
  (\bibinfo{year}{2006}), \eprint{quant-ph/0611295}.

\bibitem[{\citenamefont{Barnum et~al.}(2007)\citenamefont{Barnum, Barrett,
  Leifer, and Wilce}}]{Barnum:2006a}
\bibinfo{author}{\bibfnamefont{H.}~\bibnamefont{Barnum}},
  \bibinfo{author}{\bibfnamefont{J.}~\bibnamefont{Barrett}},
  \bibinfo{author}{\bibfnamefont{M.}~\bibnamefont{Leifer}}, \bibnamefont{and}
  \bibinfo{author}{\bibfnamefont{A.}~\bibnamefont{Wilce}},
  \bibinfo{journal}{Phys. Rev. Lett.} \textbf{\bibinfo{volume}{99}},
  \bibinfo{pages}{240501} (\bibinfo{year}{2007}), \eprint{arXiv:0707.0620v1}.

\bibitem[{\citenamefont{de~Finetti}(1993)}]{Finetti:1993}
\bibinfo{author}{\bibfnamefont{B.}~\bibnamefont{de~Finetti}},
  \emph{\bibinfo{title}{Probabilit{\`a} e Induzione----Induction and
  Probability}}, Biblioteca di STATISTICA (\bibinfo{publisher}{CLUEB, Bologna},
  \bibinfo{year}{1993}), \bibinfo{note}{a collection of de Finetti's original
  papers and English translations}.

\bibitem[{\citenamefont{de~Finetti}(1990)}]{Finetti:1990aa}
\bibinfo{author}{\bibfnamefont{B.}~\bibnamefont{de~Finetti}},
  \emph{\bibinfo{title}{Theory of Probability: A Critical Introductory
  Treatment}}, vol. \bibinfo{volume}{1 and 2} (\bibinfo{publisher}{Wiley},
  \bibinfo{year}{1990}), \bibinfo{edition}{english translation by a. smith} ed.

\bibitem[{\citenamefont{Savage}(1972)}]{Savage:1972aa}
\bibinfo{author}{\bibfnamefont{L.~J.} \bibnamefont{Savage}},
  \emph{\bibinfo{title}{The Foundation of Statistics}}
  (\bibinfo{publisher}{Dover}, \bibinfo{year}{1972}).

\bibitem[{\citenamefont{Bernardo and Smith}(2000)}]{Bernardo:2000aa}
\bibinfo{author}{\bibfnamefont{J.~M.} \bibnamefont{Bernardo}} \bibnamefont{and}
  \bibinfo{author}{\bibfnamefont{A.~F.~M.} \bibnamefont{Smith}},
  \emph{\bibinfo{title}{Bayesian Theory}}, Wiley Series in Probability and
  Statistics (\bibinfo{publisher}{Wiley}, \bibinfo{year}{2000}).

\bibitem[{\citenamefont{Hewitt and Savage}(1955)}]{Hewitt:1955}
\bibinfo{author}{\bibfnamefont{E.}~\bibnamefont{Hewitt}} \bibnamefont{and}
  \bibinfo{author}{\bibfnamefont{L.~J.} \bibnamefont{Savage}},
  \bibinfo{journal}{Trans. Amer. Math. Soc.} \textbf{\bibinfo{volume}{80}},
  \bibinfo{pages}{470} (\bibinfo{year}{1955}).

\bibitem[{\citenamefont{Hudson and Moody}(1976)}]{Hudson:1976aa}
\bibinfo{author}{\bibfnamefont{R.~L.} \bibnamefont{Hudson}} \bibnamefont{and}
  \bibinfo{author}{\bibfnamefont{G.~R.} \bibnamefont{Moody}},
  \bibinfo{journal}{Z. Wahrschein. verw. Geb.} \textbf{\bibinfo{volume}{33}},
  \bibinfo{pages}{343} (\bibinfo{year}{1976}).

\bibitem[{\citenamefont{Hudson}(1981)}]{Hudson:1981aa}
\bibinfo{author}{\bibfnamefont{R.~L.} \bibnamefont{Hudson}},
  \bibinfo{journal}{Found. Phys.} \textbf{\bibinfo{volume}{11}},
  \bibinfo{pages}{805} (\bibinfo{year}{1981}).

\bibitem[{\citenamefont{Caves et~al.}(2002{\natexlab{a}})\citenamefont{Caves,
  Fuchs, and Schack}}]{Caves:2002aa}
\bibinfo{author}{\bibfnamefont{C.~M.} \bibnamefont{Caves}},
  \bibinfo{author}{\bibfnamefont{C.~A.} \bibnamefont{Fuchs}}, \bibnamefont{and}
  \bibinfo{author}{\bibfnamefont{R.}~\bibnamefont{Schack}},
  \bibinfo{journal}{J. Math. Phys.} \textbf{\bibinfo{volume}{43}},
  \bibinfo{pages}{4537} (\bibinfo{year}{2002}{\natexlab{a}}),
  \eprint{quant-ph/0104088}.

\bibitem[{\citenamefont{Fuchs and Schack}(2004)}]{Fuchs:2004aa}
\bibinfo{author}{\bibfnamefont{C.~A.} \bibnamefont{Fuchs}} \bibnamefont{and}
  \bibinfo{author}{\bibfnamefont{R.}~\bibnamefont{Schack}}, in
  \emph{\bibinfo{booktitle}{Quantum State Estimation}}, edited by
  \bibinfo{editor}{\bibfnamefont{M.~G.~A.} \bibnamefont{Paris}}
  \bibnamefont{and}
  \bibinfo{editor}{\bibfnamefont{J.}~\bibnamefont{\v{R}eh\'{a}\v{c}ek}}
  (\bibinfo{publisher}{Springer}, \bibinfo{year}{2004}), vol.
  \bibinfo{volume}{649} of \emph{\bibinfo{series}{Lecture Notes in Physics}},
  chap.~\bibinfo{chapter}{5}, pp. \bibinfo{pages}{147--182},
  \eprint{quant-ph/0404156}.

\bibitem[{\citenamefont{Caves et~al.}(2002{\natexlab{b}})\citenamefont{Caves,
  Fuchs, and Schack}}]{Caves:2002ab}
\bibinfo{author}{\bibfnamefont{C.~M.} \bibnamefont{Caves}},
  \bibinfo{author}{\bibfnamefont{C.~A.} \bibnamefont{Fuchs}}, \bibnamefont{and}
  \bibinfo{author}{\bibfnamefont{R.}~\bibnamefont{Schack}},
  \bibinfo{journal}{Phys. Rev. A} \textbf{\bibinfo{volume}{65}},
  \bibinfo{pages}{022305} (\bibinfo{year}{2002}{\natexlab{b}}),
  \eprint{quant-ph/0106133}.

\bibitem[{\citenamefont{Fuchs}(2002)}]{Fuchs:2002aa}
\bibinfo{author}{\bibfnamefont{C.~A.} \bibnamefont{Fuchs}}
  (\bibinfo{year}{2002}), \eprint{quant-ph/0205039}.

\bibitem[{\citenamefont{Fuchs}(2003)}]{Fuchs:2003aa}
\bibinfo{author}{\bibfnamefont{C.~A.} \bibnamefont{Fuchs}},
  \bibinfo{journal}{J. Mod. Opt.} \textbf{\bibinfo{volume}{50}},
  \bibinfo{pages}{987} (\bibinfo{year}{2003}).

\bibitem[{\citenamefont{Pitowsky}(2003)}]{Pitowsky:2003aa}
\bibinfo{author}{\bibfnamefont{I.}~\bibnamefont{Pitowsky}},
  \bibinfo{journal}{Stud. Hist. Phil. Mod. Phys.}
  \textbf{\bibinfo{volume}{34}}, \bibinfo{pages}{395} (\bibinfo{year}{2003}),
  \eprint{quant-ph/0208121}.

\bibitem[{\citenamefont{Caves et~al.}(2006)\citenamefont{Caves, Fuchs, and
  Schack}}]{Caves:2006}
\bibinfo{author}{\bibfnamefont{C.~M.} \bibnamefont{Caves}},
  \bibinfo{author}{\bibfnamefont{C.~A.} \bibnamefont{Fuchs}}, \bibnamefont{and}
  \bibinfo{author}{\bibfnamefont{R.}~\bibnamefont{Schack}}
  (\bibinfo{year}{2006}), \eprint{quant-ph/0608190}.

\bibitem[{\citenamefont{Diaconis}(1977)}]{Diaconis:1977aa}
\bibinfo{author}{\bibfnamefont{P.}~\bibnamefont{Diaconis}},
  \bibinfo{journal}{Synthese} \textbf{\bibinfo{volume}{36}},
  \bibinfo{pages}{271} (\bibinfo{year}{1977}).

\bibitem[{\citenamefont{Diaconis and Freedman}(1980)}]{Diaconis:1980aa}
\bibinfo{author}{\bibfnamefont{P.}~\bibnamefont{Diaconis}} \bibnamefont{and}
  \bibinfo{author}{\bibfnamefont{D.}~\bibnamefont{Freedman}},
  \bibinfo{journal}{Ann. Probab.} \textbf{\bibinfo{volume}{8}},
  \bibinfo{pages}{745} (\bibinfo{year}{1980}).

\bibitem[{\citenamefont{Kendall}(1967)}]{Kendall:1967}
\bibinfo{author}{\bibfnamefont{D.~G.} \bibnamefont{Kendall}},
  \bibinfo{journal}{Studia Sci. Math. Hungar.} \textbf{\bibinfo{volume}{2}},
  \bibinfo{pages}{319} (\bibinfo{year}{1967}).

\bibitem[{\citenamefont{Koenig and Renner}(2005)}]{Koenig:2005aa}
\bibinfo{author}{\bibfnamefont{R.}~\bibnamefont{Koenig}} \bibnamefont{and}
  \bibinfo{author}{\bibfnamefont{R.}~\bibnamefont{Renner}},
  \bibinfo{journal}{J. Math. Phys.} \textbf{\bibinfo{volume}{46}},
  \bibinfo{pages}{122108} (\bibinfo{year}{2005}), \eprint{quant-ph/0410229}.

\bibitem[{\citenamefont{D'Cruz et~al.}(2007)\citenamefont{D'Cruz, Osborne, and
  Schack}}]{DCruz:2007}
\bibinfo{author}{\bibfnamefont{C.}~\bibnamefont{D'Cruz}},
  \bibinfo{author}{\bibfnamefont{T.~J.} \bibnamefont{Osborne}},
  \bibnamefont{and} \bibinfo{author}{\bibfnamefont{R.}~\bibnamefont{Schack}},
  \bibinfo{journal}{Phys. Rev. Lett.} \textbf{\bibinfo{volume}{98}},
  \bibinfo{pages}{160406} (\bibinfo{year}{2007}), \eprint{quant-ph/0606139}.

\bibitem[{\citenamefont{Renner}(2007)}]{Renner:2007}
\bibinfo{author}{\bibfnamefont{R.}~\bibnamefont{Renner}}
  (\bibinfo{year}{2007}), \eprint{quant-ph/0703069}.

\bibitem[{\citenamefont{Koenig and Mitchison}(2007)}]{Koenig:2007}
\bibinfo{author}{\bibfnamefont{R.}~\bibnamefont{Koenig}} \bibnamefont{and}
  \bibinfo{author}{\bibfnamefont{G.}~\bibnamefont{Mitchison}}
  (\bibinfo{year}{2007}), \eprint{quant-ph/0703210}.

\bibitem[{\citenamefont{Fuchs et~al.}(2004)\citenamefont{Fuchs, Schack, and
  Scudo}}]{Fuchs:2004ab}
\bibinfo{author}{\bibfnamefont{C.~A.} \bibnamefont{Fuchs}},
  \bibinfo{author}{\bibfnamefont{R.}~\bibnamefont{Schack}}, \bibnamefont{and}
  \bibinfo{author}{\bibfnamefont{P.~F.} \bibnamefont{Scudo}},
  \bibinfo{journal}{Phys. Rev. A} \textbf{\bibinfo{volume}{69}},
  \bibinfo{pages}{062305} (\bibinfo{year}{2004}), \eprint{quant-ph/0307198}.

\bibitem[{\citenamefont{Christandl et~al.}(2006)\citenamefont{Christandl,
  Koenig, Mitchison, and Renner}}]{Christandl:2006aa}
\bibinfo{author}{\bibfnamefont{M.}~\bibnamefont{Christandl}},
  \bibinfo{author}{\bibfnamefont{R.}~\bibnamefont{Koenig}},
  \bibinfo{author}{\bibfnamefont{G.}~\bibnamefont{Mitchison}},
  \bibnamefont{and} \bibinfo{author}{\bibfnamefont{R.}~\bibnamefont{Renner}}
  (\bibinfo{year}{2006}), \eprint{quant-ph/0602130}.

\bibitem[{\citenamefont{Mitchison}(2007)}]{Mitchison:2007}
\bibinfo{author}{\bibfnamefont{G.}~\bibnamefont{Mitchison}}
  (\bibinfo{year}{2007}), \eprint{quant-ph/0701064}.

\bibitem[{\citenamefont{Foulis and Randall}(1972)}]{Foulis:1972}
\bibinfo{author}{\bibfnamefont{D.}~\bibnamefont{Foulis}} \bibnamefont{and}
  \bibinfo{author}{\bibfnamefont{C.}~\bibnamefont{Randall}},
  \bibinfo{journal}{J. Math. Phys.} \textbf{\bibinfo{volume}{13}},
  \bibinfo{pages}{1667} (\bibinfo{year}{1972}).

\bibitem[{\citenamefont{Randall and Foulis}(1973)}]{Randall:1973}
\bibinfo{author}{\bibfnamefont{C.}~\bibnamefont{Randall}} \bibnamefont{and}
  \bibinfo{author}{\bibfnamefont{D.}~\bibnamefont{Foulis}},
  \bibinfo{journal}{J. Math. Phys.} \textbf{\bibinfo{volume}{14}},
  \bibinfo{pages}{1472} (\bibinfo{year}{1973}).

\bibitem[{\citenamefont{Greechie}(1969)}]{Greechie:1969}
\bibinfo{author}{\bibfnamefont{R.~J.} \bibnamefont{Greechie}},
  \bibinfo{journal}{Carribean Journal of Science and Mathematics}
  \textbf{\bibinfo{volume}{1}} (\bibinfo{year}{1969}).

\bibitem[{\citenamefont{Shultz}(1974)}]{Shultz:1974}
\bibinfo{author}{\bibfnamefont{F.~W.} \bibnamefont{Shultz}},
  \bibinfo{journal}{Journal of Combinatorial Theory A}
  \textbf{\bibinfo{volume}{17}}, \bibinfo{pages}{317} (\bibinfo{year}{1974}).

\bibitem[{\citenamefont{Wilce}(2000)}]{Wilce:2000}
\bibinfo{author}{\bibfnamefont{A.}~\bibnamefont{Wilce}}, in
  \emph{\bibinfo{booktitle}{Current Research in Operational Quantum Logic}},
  edited by \bibinfo{editor}{\bibfnamefont{B.}~\bibnamefont{Coecke}},
  \bibinfo{editor}{\bibfnamefont{D.}~\bibnamefont{Moore}}, \bibnamefont{and}
  \bibinfo{editor}{\bibfnamefont{A.}~\bibnamefont{Wilce}}
  (\bibinfo{publisher}{Kluwer}, \bibinfo{year}{2000}), pp.
  \bibinfo{pages}{81--114}.

\bibitem[{\citenamefont{Gleason}(1957)}]{Gleason:1957}
\bibinfo{author}{\bibfnamefont{A.~M.} \bibnamefont{Gleason}},
  \bibinfo{journal}{Journal of Mathematics and Mechanics}
  \textbf{\bibinfo{volume}{6}}, \bibinfo{pages}{885} (\bibinfo{year}{1957}).

\bibitem[{\citenamefont{Busch}(2003)}]{Busch:2003}
\bibinfo{author}{\bibfnamefont{P.}~\bibnamefont{Busch}},
  \bibinfo{journal}{Phys. Rev. Lett.} \textbf{\bibinfo{volume}{91}},
  \bibinfo{pages}{120403} (\bibinfo{year}{2003}), \eprint{quant-ph/9909073}.

\bibitem[{\citenamefont{Caves et~al.}(2004)\citenamefont{Caves, Fuchs, Manne,
  and Renes}}]{Caves:2004}
\bibinfo{author}{\bibfnamefont{C.~M.} \bibnamefont{Caves}},
  \bibinfo{author}{\bibfnamefont{C.~A.} \bibnamefont{Fuchs}},
  \bibinfo{author}{\bibfnamefont{K.}~\bibnamefont{Manne}}, \bibnamefont{and}
  \bibinfo{author}{\bibfnamefont{J.~M.} \bibnamefont{Renes}},
  \bibinfo{journal}{Found. Phys.} \textbf{\bibinfo{volume}{34}},
  \bibinfo{pages}{193} (\bibinfo{year}{2004}), \eprint{quant-ph/0306179}.

\bibitem[{\citenamefont{Pulmannov{\'a} and Wilce}(1995)}]{Pulmannova:1995}
\bibinfo{author}{\bibfnamefont{S.}~\bibnamefont{Pulmannov{\'a}}}
  \bibnamefont{and} \bibinfo{author}{\bibfnamefont{A.}~\bibnamefont{Wilce}},
  \bibinfo{journal}{Int. J. Theor. Phys.} \textbf{\bibinfo{volume}{34}},
  \bibinfo{pages}{1689} (\bibinfo{year}{1995}).

\bibitem[{\citenamefont{Gudder}(1997)}]{Gudder:1997}
\bibinfo{author}{\bibfnamefont{S.}~\bibnamefont{Gudder}},
  \bibinfo{journal}{Found. Phys.} \textbf{\bibinfo{volume}{27}},
  \bibinfo{pages}{287} (\bibinfo{year}{1997}).

\bibitem[{\citenamefont{Singer and Stulpe}(1992)}]{Singer:1992}
\bibinfo{author}{\bibfnamefont{M.}~\bibnamefont{Singer}} \bibnamefont{and}
  \bibinfo{author}{\bibfnamefont{W.}~\bibnamefont{Stulpe}},
  \bibinfo{journal}{J. Math. Phys.} \textbf{\bibinfo{volume}{33}},
  \bibinfo{pages}{131} (\bibinfo{year}{1992}).

\bibitem[{\citenamefont{Christandl and Toner}(2007)}]{Christandl:2007}
\bibinfo{author}{\bibfnamefont{M.}~\bibnamefont{Christandl}} \bibnamefont{and}
  \bibinfo{author}{\bibfnamefont{B.}~\bibnamefont{Toner}}
  (\bibinfo{year}{2007}), \eprint{arXiv:0712.0916}.

\end{thebibliography}
\end{document}